\documentclass[letterpaper, 10 pt, conference]{ieeeconf}  

\IEEEoverridecommandlockouts                              

\usepackage[T1]{fontenc}
\usepackage{microtype}
\usepackage{amsmath, amssymb, amsfonts}
\usepackage{mathtools}
\usepackage{cite}
\usepackage{tikz,pgfplots}
\usepackage{derivative}
\usepackage[caption=false]{subfig}
\usepackage{booktabs}
\usepackage[edges]{forest}
\usepackage{mathrsfs}

\usepackage{multirow}
\usepackage{makecell}

\usepackage{tikz}
\usepackage{pgfplots}   
\pgfplotsset{compat=newest}   
\usetikzlibrary{plotmarks}   
\usetikzlibrary{arrows.meta}   
\usepgfplotslibrary{patchplots}   
\usepackage{grffile}

\usepackage{amsthm}
\theoremstyle{plain}

\newtheorem{thm}{Theorem}

\theoremstyle{definition}
\newtheorem{assump}{Assumption}

\pgfplotsset{compat=newest}

\newcommand{\ie}{\emph{i.e.},}
\newcommand{\eg}{\emph{e.g.},}
\newcommand{\mb}{\mathbf}
\newcommand{\mc}{\mathcal}
\newcommand{\mbb}{\mathbb}

\newcommand{\bn}{\bar{\nu}}

\DeclarePairedDelimiter{\interiorpars}{(}{)}
\newcommand{\interior}{\operatorname{interior}\interiorpars}

\DeclarePairedDelimiter\abs{\lvert}{\rvert}%
\DeclarePairedDelimiter\norm{\lVert}{\rVert}%

\DeclareMathOperator{\E}{\mathbb{E}}
\DeclareMathOperator{\Eo}{\mathbb{E}_\nu}

\DeclareMathOperator{\Probo}{\mathbb{P}_\nu}

\DeclareMathOperator{\Cov}{Cov}
\newcommand{\mytrace}[1]{\operatorname{tr}(#1)}

\providecommand\given{}
\DeclarePairedDelimiterXPP\Eb[1]{\E}{[}{]}{}{
\renewcommand\given{  \nonscript\:
  \delimsize\vert
  \nonscript\:
  \mathopen{}
  \allowbreak}
#1
}

\DeclarePairedDelimiterXPP\Ebo[1]{\Eo}{[}{]}{}{
\renewcommand\given{  \nonscript\:
  \delimsize\vert
  \nonscript\:
  \mathopen{}
  \allowbreak}
#1
}

\DeclarePairedDelimiterXPP\Po[1]{\Probo}{[}{]}{}{
\renewcommand\given{  \nonscript\:
  \delimsize\vert
  \nonscript\:
  \mathopen{}
  \allowbreak}
#1
}

\providecommand\given{}
\DeclarePairedDelimiterXPP\Covb[1]{\Cov}{[}{]}{}{
\renewcommand\given{  \nonscript\:
  \delimsize\vert
  \nonscript\:
  \mathopen{}
  \allowbreak}
#1
}

\makeatletter
\let\oldabs\abs
\def\abs{\@ifstar{\oldabs}{\oldabs*}}
\let\oldnorm\norm
\def\norm{\@ifstar{\oldnorm}{\oldnorm*}}
\makeatother

\title{\LARGE \bf
When expectation fails: stochastic MPC of linear systems with random input losses
}

\author{Paul Trodden, Xinda Li
\thanks{The authors are with the School of Electrical and Electronic Engineering, University of Sheffield, UK. Email:
        {\tt\small\{p.trodden, xli431\}@sheffield.ac.uk}}%
}

\begin{document}

\maketitle
\thispagestyle{empty}
\pagestyle{empty}

\begin{abstract}
We consider stochastic model predictive control (MPC) for constrained linear systems subject to multiplicative binary input uncertainty, motivated by applications such as networked control with packet losses and intermittent actuation. A common approach in this setting replaces the stochastic dynamics with their expectation, yielding tractable formulations that admit standard terminal ingredients and stability guarantees in expectation. We show that such formulations can exhibit structural properties that differ fundamentally from those of deterministic MPC and may be misleading as indicators of realized closed-loop behaviour. In particular, the expected value function is not necessarily monotonic in the prediction horizon, and value function–based inner approximations of the region of attraction may deteriorate as the horizon increases. Furthermore, we establish a probabilistic comparison with certainty-equivalent (optimistic) MPC, showing that the latter can ensure a strictly positive probability of recursive feasibility in situations where stochastic MPC certifies feasibility but fails with probability one. These results highlight inherent limitations of expectation-based stochastic MPC for systems with multiplicative binary uncertainty and motivate a re-examination of how stochasticity is incorporated into constrained predictive control design for such systems.
%
%
\end{abstract}

\section{Introduction}

In this paper, we consider constrained linear systems subject to multiplicative binary uncertainty on the input, a setting that captures a range of phenomena including faulty or intermittent actuation. A particularly important instance of this setting arises in networked control systems~\cite{HNX07,SGJ22,FKW23}, where control inputs are transmitted over unreliable communication channels and may be randomly lost. A classical and seminal line of work~\cite{SSF+07,YX11,GSG+12} has shown that such effects can be captured through probabilistic transmission models, leading to fundamental limitations on achievable performance; in particular, stochastic stability of a linear system subject to random packet losses depends critically on the packet arrival rate exceeding a system-dependent threshold~\cite{SSF+07}.

The problem is more challenging still when constraints are present on the system. These arise from physical limitations, safety requirements and operational considerations, and lead naturally to model predictive control (MPC)~techniques~\cite{RR+17} as a means of achieving effective control. In the setting where the input uncertainty is probabilistic, stochastic MPC---in which control inputs are computed by optimizing an expected cost subject to constraints defined with respect to predicted behaviour~\cite{KCbook}---provides a logical and accurate means of handling the uncertainty. 

A common paradigm in stochastic MPC is the use of expected dynamics, in which random effects---here the packet losses or multiplicative uncertainties---are replaced by their mean values~\cite{HPO+17}. This yields tractable optimization problems and potentially enables the use of standard terminal ingredients for stability and feasibility~\cite{KCbook}. As a consequence, such formulations may naturally inherit Lyapunov-type guarantees in expectation and align with classical threshold-based characterizations of stochastic stability.

The use of expected dynamics in this setting, however, introduces a fundamental inconsistency: no realization of the stochastic system evolves according to these averaged dynamics. In particular, when the uncertainty is binary multiplicative, the realized system switches between two distinct modes and never coincides with its expectation. This raises an important question: what role do the structural properties of stochastic MPC formulations based on expected dynamics play in certifying the stability and performance of the actual closed-loop system? While expectation-based analysis is sufficient to establish mean stability properties, it is not clear whether this provides reliable insight into feasibility, invariance, and performance along realized trajectories.

It is important to note that many existing stochastic MPC formulations for systems with packet losses rely on assumptions or modifications that simplify analysis but substantially alter the nature of the underlying problem. These include, for example, the use of unconstrained~\cite{RQA13} or partially constrained~\cite{MCQ18} settings, the assumption of global terminal control Lyapunov functions~\cite{QN12}, and the introduction of auxiliary mechanisms such as input buffering~\cite{QN11,prabhat2018stabilizing,LKZ18} or controller add-ons and modifications~\cite{TD13,sun2019resilient, APC+22, WPS+22} to mitigate packet losses.  While such modifications can restore desirable properties such as recursive feasibility or stability in expectation, they do so by effectively reshaping the problem being solved or building mitigatory actions into the controller, rather than addressing the fundamental impact of stochasticity on the properties of the basic controller.

In this paper, we investigate these issues for constrained linear systems subject to random input losses. We compare stochastic MPC formulations based on expected dynamics with their optimistic, certainty-equivalent counterparts, which neglect stochastic effects in prediction. Our analysis reveals structural issues and discrepancies that do not arise in deterministic MPC and have important implications for the interpretation of guarantees from stochastic MPC. In particular, our main contributions are:
\begin{enumerate}
    \item We show that stochastic MPC formulations employing terminal ingredients based on expected dynamics do not inherit desirable structural properties inherent to deterministic MPC with terminal conditions. In particular, the expected value function is not necessarily monotonic in the prediction horizon.
    \item We establish that value function–based inner approximations of the region of attraction may deteriorate as the prediction horizon increases. Specifically, sublevel sets of the MPC value function are not guaranteed to enlarge with horizon length, and may even shrink owing to the propagation of variance over the horizon. This has significant implications for attempts to certify stability based on bounding the value function.
    \item We perform a probabilistic comparison between stochastic MPC and optimistic MPC, and establish a probabilistic dominance of the latter. In particular, we find that optimistic MPC ensures a strictly positive probability of the system states remaining with the controller's region of feasibility, but, in stark contrast, stochastic MPC may certify feasibility of some states which leave the region of feasibility with probability one.
\end{enumerate}

The organization of this paper is as follows. Section~\ref{sec:prob} defines the problem setting. In Section~\ref{sec:formulations}, the formulations of stochastic and optimistic MPC problems are given, both with and without terminal constraints. Section~\ref{sec:main} presents our main theoretical results, which are illustrated by examples in Section~\ref{sec:illus}. Conclusions are drawn in Section~\ref{sec:conc}.

\textit{Notation:} $\mathbb{R}_{\geq0}$ is the set of non-negative real numbers and $\mathbb{N}_{\geq0}$ the non-negative integers. $\mathbb{R}^n$ and $\mathbb{R}^{n \times m}$ are respectively the sets of real-valued $n$ vectors and $n \times m$ matrices. The generic norm of a vector $x$ is denoted by $\abs{x}$. For a matrix  $M=M^\top$, $M \succeq 0$ denotes positive semidefiniteness and $M \succ 0$ denotes positive definiteness. The spectral radius of a square matrix $M$ is noted by $\rho(M)$. The expectation of a random variable $x$ is $\Eb{x}$.

\section{Problem setting}
\label{sec:prob}

We consider the discrete-time linear time-invariant system
\begin{equation}
  (\forall k \in \mbb{N}) \ x_{k+1} = Ax_k + \nu_k Bu_{k}, 
  \label{system}%
\end{equation}
where the $x_k \in \mathbb{R}^n$ and $u_k \in \mathbb{R}^m$ are, respectively, the state and input at time $k \in \mbb{N}$, and $\nu_k$ is a binary random variable.

This models the situation where the controller and plant are connected by a communication channel that is subject to packet losses; although the controller always receives the state measurement $x_k$ at time $k$, the control input $u_k$ is not necessarily successfully transmitted to the plant actuators: a value of $\nu_k=1$ indicates successful transmission of $u_k$, while $\nu_k=0$ indicates a packet loss. It also handles a range of applications where multiplicative, binary uncertainty is present; for example, actuator faults.

The system is subject to the constraints
\begin{equation}
    (\forall k \in \mbb{N}) \ u_k \in \mathbb{U}, \quad x_k \in \mathbb{X}.
\end{equation}
We make the following assumptions about this setup.

\begin{assump}
    The matrix $A$ is unstable, \ie~$\rho(A) > 1$.
\end{assump}

\begin{assump}
    The pair $(A,B)$ is reachable.
\end{assump}

\begin{assump}
    The sequence $\{\nu_k\}$ consists of i.i.d. random variables, each drawn from a Bernoulli distribution with mean value $\bar{\nu}>0$.
\end{assump}

\begin{assump}\label{assump:measurement}
    The value of $x_k$ is known at time step $k \in \mbb{N}$.
\end{assump}

\begin{assump}
    The set $\mathbb{U}$ is compact and the set $\mathbb{X}$ is closed. Each set contains the origin in its interior.
\end{assump}

The aim is to regulate the state $x_k$ to the origin, while satisfying the constraints, despite the random binary uncertainty.

This is a simplified, basic instance of the more general control problem over lossy communication channels; it is considered in such a basic form in the current paper to maximize clarity of the exposition, findings, and implications. 

\section{Stochastic and optimistic model predictive control}
\label{sec:formulations}

In this section, we introduce two different stochastic MPC formulations---one with terminal constraints and one without---and the corresponding \emph{optimistic} MPC formulations. These formulations are standard and are presented here primarily to establish a common framework and notation for the developments that follow.

\subsection{Terminally constrained stochastic MPC}

We consider a classical formulation of a stochastic model predictive control, in which the object is to minimize an expected cost subject to the constraints. At each time step $k$, given the measured state $x_k$, a finite-horizon stochastic optimal control problem is solved, and only the first control input is applied. The optimal control problem $\mathbb{P}_{N,\bn}^s(x_k)$ at a state $x_k$ is defined as follows; for notational convenience, we write $x \coloneqq x_k$.
\begin{equation}\label{eq:smpc_cost}
    V^s_{N,\bn}(x) = \min_{\mb{u}_N} \Eb{V_N(x,\mb{u}_N, \{\nu_i\}) \given x},
\end{equation}
subject to, for $i = 0,\dots,N-1$,
\begin{subequations}
    \begin{align}
        x_0 &= x, \label{eq:init}\\
        x_{i+1} &= Ax_i + \nu_i B u_i, \\  
        u_i &\in \mbb{U},\\
        x_i &\in \mbb{X},\\
        x_N &\in \mathbb{X}_f.\label{eq:term}
    \end{align}\label{eq:smpccons}
\end{subequations}
In this problem, the index $i$ denotes the prediction time relative to the current state $x_k$, and $\{\nu_i\}_{i=0}^{N-1}$ are i.i.d. Bernoulli random variables with the same distribution as $\nu_k$. The decision variable is the sequence of controls
\begin{equation}
    \mb{u}_N \coloneqq \{ u_0, u_1, \ldots, u_{N-1}\},
\end{equation}
and the cost function is
\begin{equation}
    V_N(x,\mb{u}_N,\{\nu_i\} ) \coloneqq V_{f,\bn}(x_N) + \sum_{i=0}^{N-1} \ell(x_i,\nu_i u_i),
\end{equation}
where $V_f$ and $\ell$ are, respectively, the terminal and stage costs
\begin{align}
    V_{f,\bn}(x) &\coloneqq x^\top P_{\bn} x, \\
    \ell(x,u) &\coloneqq x^\top Q x + u^\top R u,
\end{align}
satisfying the following assumption; further assumptions on $P_{\bn}$ will be stated in due course.
\begin{assump}
    The matrices $Q, R, P_{\bn} \succ 0$.
\end{assump}

The expectation in \eqref{eq:smpc_cost} can be evaluated analytically, yielding an equivalent deterministic problem $\bar{P}^s_{N,\bar{\nu}}(x)$ (\ie~identical in minimizer and optimal value) expressed in terms of the predicted state mean and covariance. Define
\begin{align}
    \bar{x}_i &\coloneqq \Eb{x_{k+i} \given x_k = x},\\
    \Sigma_i &\coloneqq \Eb{ (x_{k+i} - \bar{x}_i) (x_{k+i} - \bar{x}_i)^\top \given x_k = x},
\end{align}
The equivalent problem  $\bar{\mbb{P}}^s_{N,\bar{\nu}}(x)$ is
\begin{multline}\label{eq:smpc_cost3}
    V_{N,\bn}^s(x) =  \min_{\mb{u}_N} V_f (\bar{x}_N) + \mytrace{P_{\bn}\Sigma_N} \\ + \sum_{i=0}^{N-1} \ell(\bar{x}_i,\sqrt{\bar{\nu}} u_i) + \mytrace{Q\Sigma_i},
\end{multline}
subject to, for $i = 0,\dots,N-1$,
\begin{subequations}
    \begin{align}
        \bar{x}_0 &= x, \\
        \Sigma_0 &= 0,\\
        \bar{x}_{i+1} &= A\bar{x}_i + \bar{\nu} B u_i, \label{eq:ave_dyn}\\
        \Sigma_{i+1} &= A \Sigma_i A^\top + \bar{\nu}(1-\bar{\nu})Bu_i u_i^\top B^\top,\\
        u_i &\in \mbb{U},\\
        \bar{x}_i &\in \mbb{X},\\
        \bar{x}_N &\in \mbb{X}_{f,\bn}.
    \end{align}
\end{subequations}
Problem $\bar{\mbb{P}}^s_{N,\bar{\nu}}(x)$ utilizes the \emph{expected} dynamics of the system to optimize for a sequence of admissible inputs that lead to the expected terminal state satisfying a designed terminal constraint. If, then, the terminal cost matrix $P_{\bn}$ and terminal set $\mbb{X}_f$ are designed to satisfy standard terminal conditions~\cite{RR+17}, then the closed-loop system inherits certain theoretical properties towards guaranteeing stochastic stability.
\begin{assump}\label{assump:P}
    The terminal cost function $V_{f,\bn}$ satisfies, for all $x \in \mbb{R}^n$,
    \begin{equation}\label{eq:Plyap}
        \Ebo{V_{f,\bn}(Ax + \nu B K_{f,\bar{\nu}}) + \ell(x, \nu K_{f,\bar{\nu}} x) \given x} \leq V_{f,\bn}(x) ,
    \end{equation}
    for some $K_{f,\bar{\nu}}$ that stabilizes $(A,\bar{\nu}B)$. 
\end{assump}
The particular manifestation of~\eqref{eq:Plyap}, having taken the necessary expectation over the random variable $\nu$, is
\begin{multline}\label{eq:Plyap2}
    (A+\bar{\nu} BK_{f,\bar{\nu}})^\top P_{\bn} (A+\bar{\nu} B K_{f,\bar{\nu}}) - P_{\bn}  \preceq \\  -\left[  Q + \bar{\nu} K_{f,\bar{\nu}}^\top R K_{f,\bar{\nu}} + \bar{\nu}(1-\bar{\nu})K_{f,\bar{\nu}}^\top B^\top P_{\bn} BK_{f,\bar{\nu}}\right],  
\end{multline}
where the additional term in $\bar{\nu}(1-\bar{\nu})$ originates from the variance of $\nu$; it is well known that for~\eqref{eq:Plyap2} to have a positive definite solution $P_{\bn}$ requires $\bar{\nu}$ to exceed a system-dependent critical value $\nu_c$~\cite{SSF+07}.
\begin{assump}\label{assump:Xf}
    The terminal set $\mbb{X}_{f,\bn}$ has $0 \in \interior{\mbb{X}_{f,\bn}}$ and satisfies, for all $x \in \mbb{X}_{f,\bn}$,
    \begin{align}
    K_{f,\bar{\nu}} x \in \mbb{U} \ \text{and} \ x &\in \mbb{X},\\
    \Ebo{Ax + \nu BK_{f,\bar{\nu}} x \given x} &\in \mbb{X}_{f,\bn}. \label{eq:Xf}
    \end{align}
\end{assump}
The particular instance of~\eqref{eq:Xf}, taking the expectation, is
\begin{equation}\label{eq:termnu}
    (\forall x \in \mbb{X}_{f,\bn}) \ (A+\bar{\nu}BK_{f,\bar{\nu}}) x \in \mbb{X}_{f,\bn}.
\end{equation}
The terminal conditions depicted in Assumptions~\ref{assump:P} and~\ref{assump:Xf} are therefore concerned with the \emph{average} system dynamics. When used in the MPC problem~$\mbb{P}^s_{N,\bn}(x)$ (or equivalently $\bar{\mbb{P}}^s_{N,\bar{\nu}}(x)$), these conditions lead to a Lyapunov decrease for the expected closed-loop system. To aid the statement of the result, define the set of states for which $\mbb{P}^s_{N,\bn}(x)$ is feasible as $\mbb{X}_{N,\bar{\nu}}$. Moreover, define the MPC control law as
\begin{equation}
    u = \kappa^s_{N,\bar{\nu}}(x) \coloneqq u^s_0(x),
\end{equation}
where $u^s_0(x)$ is the first control in the optimal sequence obtained by solving $\mbb{P}^s_{N,\bn}(x)$ or $\bar{\mbb{P}}^s_{N,\bar{\nu}}(x)$.
\begin{thm}\label{lem:stoch}
    The value function $V^s_{N,\bn}(x)$ satisfies, for all $x \in \mbb{X}_{N,\bar{\nu}}$,
    \begin{equation}\label{eq:stochlyap}
        \Ebo{ V_{N,\bn}^s(Ax + \nu B \kappa^s_{N,\bn}(x)) + \ell(x, \nu\kappa^s_{N,\bn}(x)) \given x}\leq V_{N,\bn}^s(x). 
    \end{equation}
\end{thm}

\begin{proof}
The result follows directly from standard stochastic MPC arguments, \eg~\cite{KCbook}. Specifically, since the terminal cost and set satisfy the average-dynamics Lyapunov conditions in Assumptions~\ref{assump:P} and \ref{assump:Xf}, the optimality of the MPC problem implies that the expected value function decreases along the closed-loop trajectories.
\end{proof}

Although Theorem~\ref{lem:stoch} establishes that the value function decreases in expectation along the closed-loop trajectories, it does \emph{not} imply recursive feasibility of the controller.
In particular, feasibility of problem $\mathbb{P}_{N,\bn}^s(x_k)$ is not guaranteed to be preserved under the realized multiplicative-noise dynamics, and it is indeed possible that $Ax + \nu B \kappa^s_{N,\bn}(x)$ leaves the set $\mbb{X}_{N,\bar{\nu}}$---and so $V_{N,\bn}^s(Ax + \nu B \kappa^s_{N,\bn}(x))$ is then undefined---with non-zero probability.

\subsection{Terminally unconstrained stochastic MPC}

The role of the terminal constraint is to induce stability and recursive feasibility. As just discussed, however, when input losses are present, terminal sets designed to be invariant for the expected dynamics do not ensure feasibility or descent of the value function along realized trajectories. To isolate the effect of the stochastic dynamics from feasibility artifacts induced by terminal constraints, we therefore now consider a formulation, $\mbb{P}^u_{N,\bar{\nu}}(x)$ at a state $x=x_k$, that omits the terminal constraint:
\begin{equation}
    V^u_{N,\bn}(x) = \min_{\mb{u}_N} \Eb{V^\beta_N(x,\mb{u}_N, \{\nu_i\}) \given x},
\end{equation}
subject to, for $i = 0,\dots,N-1$,
\begin{subequations}
    \begin{align}
        x_0 &= x,  \\
        x_{i+1} &= Ax_i + \nu_i B u_i, \\  
        u_i &\in \mbb{U},\label{eq:input}\\
        x_i &\in \mbb{X},
    \end{align}
\end{subequations}
where now
\begin{equation}
    V^\beta_N(x,\mb{u}_N,\{\nu_i\} ) \coloneqq \beta V_{f,\bn}(x_N) + \sum_{i=0}^{N-1} \ell(x_i,\nu_i u_i).
\end{equation}
Taking the expectation over $\{\nu_{k+i}\}_{i=0}^{N-1}$, the prediction model is still the average dynamics~\eqref{eq:ave_dyn} but now there is no requirement that $\bar{x}_N \in \mbb{X}_{f,\bar{\nu}}$. An equivalent form of $\mbb{P}^u_{N,\bn}(x)$ (in the sense of having the same minimizer and value) is therefore $\bar{\mbb{P}}^u_{N,\bar{\nu}}(x)$, defined as
\begin{multline}\label{eq:smpc_cost2}
{V}^u_{N,\bn} (x) = \min_{\mb{u}_N} \beta V_{f,\bn} (\bar{x}_N) + \beta \mytrace{P_{\bn}\Sigma_N} \\ + \sum_{i=0}^{N-1} \ell(\bar{x}_i,\sqrt{\bar{\nu}} u_i)  + \mytrace{Q\Sigma_i},
\end{multline}
subject to, for $i = 0,\dots,N-1$,
\begin{subequations}
    \begin{align}
        \bar{x}_0 &= x, \label{eq:uunit}\\
        \Sigma_0 &= 0, \\
        \bar{x}_{i+1} &= A\bar{x}_i + \bar{\nu} B u_i, \label{eq:udyn}\\
        \Sigma_{i+1} &= A \Sigma_i A^\top + \bar{\nu}(1-\bar{\nu})Bu_i u_i^\top B^\top,\\
        u_i &\in \mbb{U},\label{eq:uconsmy}\\
        x_i &\in \mbb{X}.\label{eq:xconsmy}
    \end{align}
\end{subequations}
In these problems, $\beta \geq 1$ is an additional weighting applied to the terminal cost; it is well known~\cite{LAS+06,RR+17} that, in a deterministic setting, stability of the closed-loop system can be guaranteed despite the omission of a terminal constraint by imposing a weighted terminal cost function $\beta V_{f,\bar{\nu}}, \beta \geq 1$; in our setting, this means that $\beta V_{f,\bn}(x) = x^\top (\beta P_{\bn}) x$ satisfies
\begin{multline}~\label{eq:ly2}
    (A+\bar{\nu} BK_{f,\bar{\nu}})^\top (\beta P_{\bn}) (A+\bar{\nu} B K_{f,\bar{\nu}}) - \beta P_{\bn}  \preceq \\- \left[ Q + \bar{\nu} K_{f,\bar{\nu}}^\top R K_{f,\bar{\nu}} + \bar{\nu}(1-\bar{\nu})K_{f,\bar{\nu}}^\top B^\top (\beta P_{\bn}) BK_{f,\bar{\nu}}\right].  
\end{multline}
It is easily shown that satisfaction of this inequality with $\beta = 1$, \ie~\eqref{eq:Plyap2}, ensures satisfaction for all $\beta > 1$.

Denote the control law resulting from this problem as
\begin{equation}
    u = {\kappa}^u_{N,\bar{\nu}}(x) = u^u_0(x),
\end{equation}
where $u^u_0(x)$ is the first control in the optimal sequence obtained by solving  ${\mbb{P}}^u_{N,\bn}(x)$  or, equivalently, $\bar{\mbb{P}}^u_{N,\bar{\nu}}(x)$. This variant of the problem then inherits a Lyapunov decrease similar to~\eqref{eq:stochlyap} for states in a subset $\Gamma_{N,\bar{\nu}}$ of the state space. In particular, define
\begin{equation}
    \Gamma_{N,\bar{\nu}} \coloneqq \{ x: {V}_{N,\bn}^{u}(x) \leq Nd_{\bar{\nu}} + \beta c_{\bar{\nu}} \},
\end{equation}
where $c_{\bar{\nu}} > 0$ and $d_{\bar{\nu}} > 0$ are constants such that
\begin{equation}\label{eq:cconst}
    \Omega_{\bar{\nu}} \coloneqq \{ x : V_{f,\bn}(x) \leq c_{\bar{\nu}} \} \subseteq \mbb{X}_{f,\bar{\nu}},
\end{equation}
and
\begin{equation}\label{eq:dconst}
    (\forall x \not\in \Omega_{\bar{\nu}}, u \in \mbb{U}) \ \ell(x,u)  > d_{\bar{\nu}}.
\end{equation}
\begin{thm}
    The value function ${V}^u_{N}(x)$ satisfies, for all $x \in \Gamma_{N,\bar{\nu}}$,
    \begin{equation}\label{eq:stochlyap1}
        \Ebo{ V_N^u(Ax + \nu B \kappa^u_{N,\bn}(x)) + \ell(x, \nu\kappa^u_{N,\bn}(x)) \given x}\leq V_{N,\bn}^u(x). 
    \end{equation}
    Therefore, for all $x \in \Gamma_{N,\bar{\nu}}$,
    \begin{equation}
        Ax + \bar{\nu}B{\kappa}^u_{N,\bar{\nu}}(x) \in \Gamma_{N,\bar{\nu}}.
    \end{equation}
\end{thm}
\begin{proof}
    It is easily shown that problem~$\bar{\mbb{P}}^u_{N,\bar{\nu}}(x)$ is equivalent to a deterministic problem in which the variance terms are eliminated and replaced by time-varying input weights:
    \begin{equation}
    {V}^u_{N} (x) = \min_{\mb{u}_N} \beta \bar{x}_N^\top P_{\bn} \bar{x}_N + \sum_{i=0}^{N-1} \bar{x}_i^\top Q \bar{x}_i + \bar{u}_i^\top \tilde{R}^\beta_{N,i} \bar{u}_i
    \end{equation}
    subject to~\eqref{eq:uunit}, \eqref{eq:udyn},~\eqref{eq:uconsmy} and~\eqref{eq:xconsmy}, where the input penalty
    \begin{equation}\label{eq:Rt}
        \tilde{R}^\beta_{N,i} \coloneqq \bar{\nu} R + \Delta^\beta_{N,i},
    \end{equation}
    with, for $i=0,\dots,N-1$,
    \begin{equation}
        \Delta^\beta_{N,i} \coloneqq \bar{\nu}(1-\bar{\nu}) B^\top S^\beta_{N,i} B,
    \end{equation}
    and
    \begin{equation}\label{eq:St}
        S^\beta_{N,i} \coloneqq \beta(A^{N-1-i})^\top P_{\bn} A^{N-1-i} + \sum_{j=0}^{N-1-i} (A^j)^\top Q A^j.
    \end{equation}
    Consider any $x \in \Gamma_{N,\bar{\nu}}$ and the optimal $\mb{u}_{N,\bn}^u(x)$. Since $\tilde{R}_i \succ 0$ for $i=0,\dots,N-1$, the optimal stage cost $(\bar{x}_i^u)^\top Q \bar{x}_i^u + ({u}_i^u)^\top \tilde{R}^\beta_{N,i} {u}_i^u \geq (\bar{x}_i^u)^\top Q \bar{x}_i^u$ for all $i=0,\dots,N-1$; moreover, $(\bar{x}_N^u)^\top P_{\bn} \bar{x}_N^u \geq (\bar{x}_N^u)^\top Q \bar{x}_N^u$. However, since $x \in \Gamma_{N,\bar{\nu}}$ then $V_{N,\bn}^u(x) \leq N d_{\bar{\nu}} + \beta c_{\bar{\nu}}$ and therefore, using the arguments from~\cite[Lemma 2.40]{RR+17}, it must hold that $\bar{x}_N^u \in \mathbb X_{f,\bar\nu}$.

    Now consider the standard tail control sequence
    \begin{equation}
        \tilde{\mb{u}}^u_N =
\{{u}_1^u(x),\dots,u_{N-1}^u(x),K_{f,\bar\nu}\bar{x}_N^u(x)\}.
    \end{equation}
Because $\bar{x}_N^u \in \mathbb X_{f,\bar\nu}$ and $\mathbb X_{f,\bar\nu}$ is positively invariant under the terminal control law, this sequence is feasible for the problem at the expected successor state $Ax+\bn B u^u_0(x)$. Evaluating the cost of this feasible sequence and using the terminal Lyapunov inequality yields~\eqref{eq:stochlyap1}. Finally, since $V_{N,\bn}^u(x_{k+1}) \le V_{N,\bn}^u(x_k)$ in expectation and
$\Gamma_{N,\bar\nu}$ is a sublevel set of $V_{N,\bn}^u$, it follows that $Ax + \bar\nu B \kappa^u_{N,\bar\nu}(x) \in \Gamma_{N,\bar\nu}$.
\end{proof}

\subsection{Optimistic MPC}

Finally, we define \emph{optimistic} variants of the preceding problems as follows. Let 
${\mbb{P}}^s_{N,1}(x)$ be a counterpart to problem ${\mbb{P}}^s_{N,\bar{\nu}}(x)$, employing the nominal system dynamics in the prediction model:
\begin{equation}
{\mbb{P}}^s_{N,1}(x)\colon    {V}^s_{N,1} (x) = \min_{\mb{u}_N} V_{f,1} ({x}_N) + \sum_{i=0}^{N-1} \ell({x}_i,u_i),
\end{equation}
subject to, for $i = 0,\dots,N-1$,
\begin{subequations}
    \begin{align}
        {x}_0 &= x, \\
        {x}_{i+1} &= A{x}_i + B u_i,\\
        u_i &\in \mbb{U},\\
        x_i &\in \mbb{X},\\
        x_N &\in \mbb{X}_{f,1}.
    \end{align}
\end{subequations}

This is a certainty equivalent formulation of problem ${\mathbb{P}}^s_{N,\bar{\nu}}(x)$ using $\bar{\nu}=1$, and therefore just a conventional constrained MPC formulation. This optimistic certainty is carried through to the design of the terminal cost; in the previous section, we designed a terminal cost using the average dynamics, which provides an expected Lyapunov decrease of the terminal cost, but does not provide \emph{pathwise} guarantees. Consequently, in this section we propose the use of the optimistic terminal cost satisfying the following assumption.
\begin{assump}\label{assump:P2}
    The terminal cost function $V_{f,1}$ satisfies, for all $x \in \mbb{R}^n$,
    \begin{equation}\label{eq:Plyapo}
        V_{f,1}(Ax + B K_{f,1}) + \ell(x, K_{f,1} x) \leq V_{f,1}(x) ,
    \end{equation}
    for some $K_{f,1}$ that stabilizes $(A,B)$. 
\end{assump}
The particular instance of~\eqref{eq:Plyapo} for $V_{f,1}(x)=x^\top P_1 x$ is
\begin{equation}\label{eq:Plyap3}
    (A+ BK_{f,1})^\top P_1 (A+ B K_{f,1}) + Q +  K_{f,1}^\top R K_{f,1} \preceq P_1 .
\end{equation}
Likewise, $\mbb{X}_{f,1}$ is a set satisfying the following assumption.
\begin{assump}\label{assump:Xf1}
    The set $\mbb{X}_{f,1}$ has $0 \in \interior{\mbb{X}_{f,1}}$ and satisfies, for all $x \in \mbb{X}_{f,1}$,
    \begin{align}
    K_{f,1} x \in \mbb{U} \ \text{and} \ x &\in \mbb{X},\\
    (A + BK_{f,1})x &\in \mbb{X}_{f,1}. 
    \end{align}
\end{assump}
That is, $V_{f,1}$ and $\mbb{X}_{f,1}$ are standard terminal ingredients designed for the deterministic, optimistic dynamics under the stabilizing terminal control law $K_{f,1}$. Crucially,~\eqref{eq:Plyap3} is the standard Lyapunov inequality, omitting the additional variance term that appears in~\eqref{eq:Plyap2}. Thus,~\eqref{eq:Plyap3} has a solution for \emph{all} $\bar{\nu} \in (0,1]$ and not just  $\bar{\nu} > \nu_c$ because it is, in fact, independent of $\bar{\nu}$. Whilst this might seem at odds with the probabilistic setting of the paper and stochastic nature of the system, it should be noted that it makes the MPC synthesis independent of $\bar{\nu}$ and more generally the random variable $\nu$. Nonetheless, the behaviour of the closed-loop system remains stochastic. In fact, as we shall see, the probabilistic stability properties of the closed-loop system depend critically on whether the terminal cost is matched to the realized dynamics.

We may also define a terminally unconstrained version of this problem, ${\mbb{P}}^u_{N,1}(x)$, in which the terminal constraint is omitted and the terminal cost is weighted by $\beta \geq 1$; this is an optimistic counterpart to ${\mathbb{P}}^u_{N,\bar{\nu}}(x)$ using $\bar{\nu}=1$. In both problems ${\mbb{P}}^s_{N,1}(x)$ and ${\mbb{P}}^u_{N,1}(x)$, the variance associated with the predicted trajectory is entirely neglected and, thus, the input weight at each prediction step $i$ is just $R$ and not $\tilde{R}_i$.

Denote the control law associated with ${\mbb{P}}^s_{N,1}(x)$ as
\begin{equation}
    u = \kappa^s_{N,1}(x) \coloneqq u^s_0(x),
\end{equation}
and likewise denote the control law associated with ${\mbb{P}}^u_{N,1}(x)$ as
\begin{equation}
    u = \kappa^u_{N,1}(x) \coloneqq u^u_0(x).
\end{equation}
The region of feasibility for $\bar{\mbb{P}}^s_{N,1}(x)$ is $\mbb{X}_{N,1}$ and it is well known that, for a state $x$ within this,
\begin{equation}
    {V}^s_{N,1} (Ax+B\kappa^s_{N,1}(x)) - {V}^s_{N,1}(x) \leq - \ell(x,\kappa^s_{N,1}(x)).
\end{equation}
Likewise, define a set $\Gamma_{N,1}$ as a counterpart to $\Gamma_{N,\bar{\nu}}$, namely 
\begin{equation}
    \Gamma_{N,1} \coloneqq \{ x: \bar{V}_{N,1}^{u}(x) \leq N d_{1} + \beta c_{1}\},
\end{equation}
where $c_{1} > 0$ and $d_{1} > 0$ are constants such that
\begin{equation}
    \Omega_{1} \coloneqq \{ x : V_{f,1}(x) \leq c_{1} \} \subseteq \mbb{X}_{f,1},
\end{equation}
and
\begin{equation}
    (\forall x \not\in \Omega_{1}, u \in \mbb{U}) \ \ell(x,u) > d_{1},
\end{equation}
Then, for all $x \in \Gamma_{N,1} $,
\begin{equation}
    {V}^u_{N,1} (Ax+B\kappa^u_{N,1}(x)) - {V}^u_{N,1}(x) \leq - \ell(x,\kappa^u_{N,1}(x)).
\end{equation}

These are standard results for deterministic MPC~\cite{RR+17,LAS+06}. Our aim in the remainder of the paper is, however, to analyse and characterize properties of, and relations between, the sets $\mbb{X}_{N,1}$, $\mbb{X}_{N,\bn}$, $\Gamma_{N,1}$ and $\Gamma_{N,\bn}$.

\section{Main results}
\label{sec:main}

\subsection{Monotonicity of the value function}

It is well known that in deterministic MPC with terminal ingredients the value function has the desirable property of monotonicity in horizon length: for all $x \in \mbb{X}_N$, $V_{N+1}(x) \leq V_N(x)$. Our first result is to show that the same property does not necessarily hold, in expectation, for stochastic MPC of linear systems with multiplicative binary input uncertainty.
\begin{thm}\label{thm:mono}
    There is no guaranteed ordering between $V^s_{N,\bar{\nu}}(x)$ and $V^s_{N-1,\bar{\nu}}(x)$.
\end{thm}

\begin{proof}
    Consider any $x \in \mbb{X}_{N-1,\bn}$ and denote the optimal solution to $\mbb{P}^s_{N-1,\bn}(x)$ as 
    \begin{equation}
        \mb{u}_{N-1,\bn}^s(x) = \{ u^s_{\bn}(0),\dots,u^s_{\bn}(N-2)\}.
    \end{equation}
    with cost
    \begin{multline}
        V^s_{N-1,\bn}(x) = (\bar{x}^s_{N-1})^\top P_{\bn} \bar{x}^s_{N-1} \\+ \sum_{i=0}^{N-2} (\bar{x}^s_i)^\top Q \bar{x}^s_i + (\bar{u}^s_i)^\top \tilde{R}^1_{N-1,i} \bar{u}^s_i.
    \end{multline}
    Since $\bar{x}^s_{N-1}(x) \in \mbb{X}_{f,\bn}$, a feasible solution to $\mbb{P}^s_{N,\bn}(x)$ is
     \begin{equation}
        \tilde{\mb{u}}_{N,\bn}^s(x) = \{ u^s_{\bn}(0),\dots,u^s_{\bn}(N-2),K_{f,\bn}x^s_{N-1} \},
    \end{equation}
    with cost
\begin{equation*}
\begin{split}
\tilde{V}_N^s(x) &= (\bar{x}^s_{N-1})^\top(A+\bn BK_{f,\bn})^\top  P_{\bn} (A+\bn BK_{f,\bn}) \bar{x}^s_{N-1} \\ &\quad + (\bar{x}^s_{N-1})^\top (Q + K_{f,\bar{\nu}}^\top \tilde{R}^1_{N,N-1} K_{f,\bn} ) \bar{x}^s_{N-1}  \\&\quad + \sum_{i=0}^{N-2} (\bar{x}^s_i)^\top Q \bar{x}^s_i + (\bar{u}^s_i)^\top \tilde{R}^1_{N,i} \bar{u}^s_i.
\end{split}
\end{equation*}
Now, using~\eqref{eq:Rt}--\eqref{eq:St} with $\beta = 1$, 
\begin{equation}
    \tilde{R}^1_{N,N-1} = \bar{\nu} R + \bn(1-\bn)B^\top(P_{\bn} + Q)B.
\end{equation}
Hence, using~\eqref{eq:Plyap},
\begin{equation*}
    \begin{split}
\tilde{V}_N^s(x) &\leq V^s_{N-1,\bn}(x) +  \sum_{i=0}^{N-2}  (\bar{u}^s_i)^\top [\tilde{R}^1_{N,i} - \tilde{R}^1_{N-1,i}] \bar{u}^s_i \\ &\quad + \bn(1-\bn)(\bar{x}^s_{N-1})^\top K_{f,\bar{\nu}}^\top B^\top Q B K_{f,\bn}  \bar{x}^s_{N-1}, 
\end{split}
\end{equation*}
which, even though $ V^s_{N,\bn}(x) \leq \tilde{V}_N^s(x)$, does not guarantee that $\tilde{V}_N^s(x) \leq \tilde{V}_{N-1}^s(x)$.
\end{proof}
The result here clearly exposes how the increase in the weights $\tilde{R}^1_{N,i}$ and the additional terminal stage term dependent on $Q$ may increase $V^s_{N,\bn}(x)$ compared to $V^s_{N,\bn}(x)$. In fact, 
\begin{multline}
\tilde{R}^1_{N,i} - \tilde{R}^1_{N-1,i} 
= \bar{\nu}(1-\bar{\nu}) B^\top 
\left[ (A^{N-1-i})^\top Q A^{N-1-i} + \right. \\ \left.  (A^{N-2-i})^\top (A^\top P_{\bar{\nu}} A - P_{\bar{\nu}}) A^{N-2-i} \right] B.
\end{multline}
which, for unstable $A$, is generally largest (in the sense of induced norm or eigenvalues) for $i=0$, and scales exponentially with $N$. Thus, $V_{N,\bn}^s(x)$ can increase substantially with horizon length, as a consequence of the backwards propagation of the variance terms over a longer horizon. An example of this is given in Sec.~\ref{sec:illus}.

\subsection{Conservatism and ordering of $\Gamma_{N,\bar{\nu}}$}

The set $\Gamma_{N,\bar{\nu}}$ provides a value function-based inner approximation to  $\mathbb{X}_{N,\bar{\nu}}$, the expected region of attraction for the controller incorporating terminal constraints. In deterministic MPC, the monotonicity property of the value function means that $\Gamma_{N+1,1} \supseteq \Gamma_{N,1}$. Therefore, $\Gamma_{N,1}$ provides a progressively improving inner estimate of $\mbb{X}_{N,1}$.

In general, the estimate that $\Gamma_{N,\bar{\nu}}$ provides is inherently conservative since the lower bound on the stage cost, ${d}_{\bar{\nu}}$, which has to hold for all $u \in \mbb{U}$, typically ignores the contribution of the input terms to the whole cost.  An interesting consequence of~Theorem~\ref{thm:mono} is that the estimate may also be \emph{deteriorating} rather than improving with horizon length. To aid the statement of the following, define $U^s_{N,\bar{\nu}}(x)$ as the minimum possible expected input cost associated with transferring the averaged system to $\Omega_{\bn}$ in $N$ steps:
\begin{equation}
    U^\beta_{N,\bar{\nu}}(x) = \min \sum_{i=0}^{N-1} u_i^\top \tilde{R}^\beta_{N,i} u_i \ \text{subject to~\eqref{eq:smpccons}.}
\end{equation}

\begin{thm}
    There is no guaranteed ordering between $\Gamma_{N,\bar{\nu}}$ and $\Gamma_{N-1,\bar{\nu}}$. Moreover, suppose $x \in \Gamma_{N-1,\bn} \setminus \Omega_{\bn}$. If
    \begin{equation}\label{eq:costlower}
        U^\beta_{N,\bar{\nu}}(x) > Nd_{\bn} + \beta c_{\bn} ,
    \end{equation}
    then $x \not\in \Gamma_{N,\bn}$.
\end{thm}

\begin{proof}
    If $x \in \Gamma_{N-1,\bn}$ then, by definition, $V^u_{N-1,\bn}(x) \leq (N-1)d_{\bar{\nu}} + \beta c_{\bn}$. For $x$ to lie in $\Gamma_{N,\bn}$, it is required that $V^u_{N,\bn}(x) \leq Nd_{\bar{\nu}} + \beta c_{\bn}$ with $x_N^u(x) \in \Omega_{\bn}$. A lower bound for $V^u_{N,\bn}(x)$ is $U^\beta_{N,\bar{\nu}}(x)$. Therefore, $V^u_{N,\bn}(x) > Nd_{\bar{\nu}} + \beta c_{\bn}$ if~\eqref{eq:costlower} holds.
\end{proof}

The result can be sharpened when $d_{\bn}$ is estimated in the usual way that omits the cost of control inputs. In particular, since $\mbb{U}$ contains the origin, $R$ is positive definite, and $\mytrace{Q\Sigma} \geq 0$, the bound~\eqref{eq:dconst} is satisfied for any $d_{\bn}$ that satisfies
\begin{equation}
     (\forall x \not\in \Omega_{\bar{\nu}}) \ x^\top Q x > d_{\bar{\nu}}.
\end{equation}
The smallest $x^\top Q x$ occurs on the boundary $x^\top P x = c_{\bar{\nu}}$; thus $d_{\bar{\nu}} = \min \{ x^\top Q x : x^\top P x = c_{\bar{\nu}}\}$, which has the solution
\begin{equation}
     d_{\bar{\nu}} = c_{\bar{\nu}} \lambda_{\min}(P^{-1}Q).
\end{equation}
In this case, the lower bound on $V^s_{N,\bn}(x)$ is improved to $d_{\bn} + U^s_{N,\bar{\nu}}(x)$, and consequently the sufficient condition for $x \not\in \Gamma_{N,\bn}$ becomes 
    \begin{equation}
        U^\beta_{N,\bar{\nu}}(x) >  (N-1)d_{\bn}+ \beta c_{\bn}.
    \end{equation}
It can be shown that $\tilde{R}^\beta_{N,i}$, which governs the size of $U^\beta_{N,\bar{\nu}}(x)$, may be written as
\begin{equation*}
\tilde{R}^\beta_{N,i}
=
\tilde{R}^1_{N,i}
+
(\beta - 1)\,\bar{\nu}(1-\bar{\nu})\,
B^\top (A^{N-1-i})^\top P_{\bar{\nu}} A^{N-1-i} B.
\end{equation*}
Thus, $\tilde{R}^\beta_{N,i}$ depends affinely on $\beta$, with a positive semidefinite increment proportional to $\beta-1$. In contrast, its dependence on the horizon length $N$ is nonlinear and, for unstable $A$, grows exponentially (in the sense of induced norm or eigenvalues) when $A$ is unstable.

The natural conclusion is that $\Gamma_{N,\bn}$ forms a progressively more conservative estimate of the expected region of attraction as the horizon is increased; thus, in contrast to the deterministic case, this approach does not offer an attractive means of systematically enlarging inner approximations of the region of attraction. Taken together with the lack of monotonicity, this indicates that stochastic MPC formulations based on the expected dynamics of the system with binary input uncertainty may exhibit fundamentally undesirable structural properties.

\subsection{Probabilistic dominance of optimistic MPC}

The preceding discussion has focused on structural properties of stochastic MPC formulations based on expected dynamics. We now highlight a more fundamental limitation of this approach. The predicted trajectories used in such formulations correspond to the averaged system, yet no realization of~\eqref{system} evolves according to these dynamics, since $\nu_k \in \{0,1\}$ and $\nu_k \neq \bar{\nu}$. As a consequence, while the stochastic formulation guarantees invariance of $\mbb{X}_{N,\bn}$ with respect to the average closed-loop dynamics, neither realized state update,
\begin{equation}
    x_{k+1}=Ax_k \quad \text{nor} \quad x_{k+1}=Ax_k+B{\kappa}_{N,\bn}^s(x_k),
\end{equation}
is guaranteed to remain in $\mbb{X}_{N,\bn}$.

In contrast, the optimistic MPC formulation—corresponding to assuming $\bar{\nu}=1$—admits predicted trajectories that are consistent with realizations of the system. As shown in this section, this leads to a form of probabilistic dominance of optimistic MPC over stochastic MPC with respect to invariance and recursive feasibility.

In what follows, the sets $\mc{X}_{1}$ (and $\mc{X}_{\bar{\nu}}$) denote either $\mbb{X}_{N,1}$ or $\Gamma_{N,1}$ (resp. $\mbb{X}_{N,\bar{\nu}}$ or $\Gamma_{N,\bar{\nu}}$), depending on whether the relevant formulation includes or omits a terminal constraint. Likewise, $\kappa_{N,1}$ ($\kappa_{N,\bn}$) denotes the corresponding control law, $\kappa^s_{N,1}$ or $\kappa^u_{N,1}$ ($\kappa^s_{N,\bn}$ or $\kappa^u_{N,\bn}$). We refrain from making a general comparison between $\mc{X}_{1}$ and $\mc{X}_{\bn}$ because the stabilizing terminal gains used in their specification, respectively $K_{f,1}$ and $K_{f,\bar{\nu}}$, are permitted to be different. In fact, there is considerably more flexibility in design of the former since only~\eqref{eq:Plyapo} and not~\eqref{eq:Plyap} must be met. Nonetheless, this design flexibility must be wielded carefully, as its impact on the stochastic properties of the closed-loop system is not immediately clear. On the positive side, as we shall see next, optimistic MPC offers a form of probabilistic dominance that provides a more meaningful connection between predicted and realized trajectories.


\begin{thm}\label{thm:D}
For all $x \in \mc{X}_{1}$, the closed-loop state under the optimistic policy satisfies
\begin{equation}
    \Po{Ax_k + \nu_k B \kappa_{N,1}(x_k)\in\mc{X}_{1}\given x_k=x} \geq \bar{\nu}.
\end{equation}
Conversely, for the stochastic MPC formulation based on expected dynamics, \emph{either} there exists a nonempty set
\begin{equation}
    \mathcal{D} \subset \mc{X}_{\bar{\nu}},
\end{equation}
such that, for all $x\in\mathcal{D}$,
\begin{equation}
\Po{Ax_k + \nu_k B \kappa_{N,\bar{\nu}}(x_k)\in\mc{X}_{\bar{\nu}}\given x_k=x} = 0.
\end{equation}
\emph{or}, for all $x \in \mc{X}_{\bn}$,
\begin{equation}
    \Po{Ax_k + \nu_k B \kappa_{N,\bar{\nu}}(x_k)\in\mc{X}_{\bar{\nu}}\given x_k=x} \geq \min\{\bn,1-\bn\}.
\end{equation}
\end{thm}

\begin{proof}
The realized system, under the control law $u = {\kappa}_{N,\bar{\nu}}(x)$, has exactly two possibilities for the successor to a state $x$:
    \begin{equation}
        Ax +  B {\kappa}_{N,\bar{\nu}}(x) \quad \text{or} \quad Ax,
    \end{equation}
    the former with probability $\bar{\nu}$ and the latter with probability $1-\bar{\nu}$.  Define the largest subset of $\mc{X}_{\bar{\nu}}$ such that the state remains within  $\mc{X}_{\bar{\nu}}$ under \emph{successful} realizations of the variable $\nu_k\in\{0,1\}$ as
    \begin{equation*}
    \mathcal{S}_{\bar{\nu}} \coloneqq \{ x \in \mc{X}_{\bar{\nu}} : Ax +  B {\kappa}_{N,\bar{\nu}}(x) \in \mc{X}_{\bar{\nu}}\}.
    \end{equation*}
    Since $\mc{X}_{\bar{\nu}}$  is invariant for the averaged system $x^+ = Ax +  \bar{\nu}B {\kappa}^u_{N,\bar{\nu}}(x)$ but \emph{not necessarily} for the realized system $x^+=Ax +  B {\kappa}_{N,\bar{\nu}}(x)$, then $\mathcal{S}_{\bar{\nu}}\subseteq \mc{X}_{\bar{\nu}}$.
 The set $\mathcal{S}_{\bar{\nu}}$ is non-empty because it includes the point $x=0$ and an open neighbourhood around it. It follows that
 \begin{equation}
     \Po{Ax_k + \nu_k B \kappa_{N,\bar{\nu}}(x_k)\in\mc{X}_{\bar{\nu}}\given x_k \in \mathcal{S}_{\bar{\nu}}} \geq \bar{\nu}.
 \end{equation}
 Likewise, define the largest subset of $\mc{X}_{\bar{\nu}}$ such that the state remains within  $\mc{X}_{\bar{\nu}}$ under \emph{unsuccessful} realizations:
\begin{equation}
    \mc{F}_{\bar{\nu}}  =  \{ x \in \mc{X}_{\bar{\nu}}  : Ax \in \mc{X}_{\bar{\nu}}\} \subset \mc{X}_{\bar{\nu}}.
\end{equation}
The set is non-empty because it includes the point $x=0$ and an open neighbourhood around it.
It follows that
 \begin{equation}
     \Po{Ax_k + \nu_k B \kappa_{N,\bar{\nu}}(x_k)\in\mc{X}_{\bar{\nu}}\given x_k \in \mc{F}_{\bar{\nu}}} \geq 1-\bar{\nu}.
 \end{equation}
It follows from these constructions that
\begin{align*}
    \Po{x_{k+1}\in\mc{X}_{\bar{\nu}}\given x_k \in \mathcal{S}_{\bar{\nu}} \cap \mathcal{F}_{\bar{\nu}}} &= 1, \\
    \Po{x_{k+1}\in\mc{X}_{\bar{\nu}}\given x_k \in \mathcal{S}_{\bar{\nu}} \setminus  \mathcal{F}_{\bar{\nu}}} &=  \bar{\nu}, \\
    \Po{x_{k+1}\in\mc{X}_{\bar{\nu}}\given x_k \in \mathcal{F}_{\bar{\nu}} \setminus  \mathcal{S}_{\bar{\nu}}} &= 1-\bar{\nu}, \\
    \Po{x_{k+1}\in\mc{X}_{\bar{\nu}}\given x_k \not\in \mathcal{S}_{\bar{\nu}} \cup \mathcal{F}_{\bar{\nu}}} &= 0.
\end{align*}
where $x_{k+1} = Ax_k + \nu_k B \kappa_{N,\bar{\nu}}(x_k)$. Since $\mathcal{S}_{\bar{\nu}}, \mathcal{F}_{\bar{\nu}} \subseteq \mc{X}_{\bar{\nu}}$ then $\mathcal{S}_{\bar{\nu}} \cup \mathcal{F}_{\bar{\nu}} \subseteq \mc{X}_{\bar{\nu}}$. Moreover, $\mathcal{S}_{\bar{\nu}} \cup \mathcal{F}_{\bar{\nu}}$ is non-empty. Therefore, \emph{either} $\mathcal{S}_{\bar{\nu}} \cup \mathcal{F}_{\bar{\nu}} = \mc{X}_{\bar{\nu}}$, in which case 
\begin{equation}
\Po{Ax_k + \nu_k B \kappa_{N,\bar{\nu}}(x_k)\in\mc{X}_{\bar{\nu}}\given x_k=x} \geq \min\{\bn,1-\bn\} > 0,
\end{equation}
\emph{or} $\mathcal{S}_{\bar{\nu}} \cup \mathcal{F}_{\bar{\nu}} \subsetneq \mc{X}_{\bar{\nu}}$, in which case there exists a non-empty subset of $\mc{X}_{\bar{\nu}}$, in particular $\mc{D} = \mc{X}_{\bar{\nu}} \setminus (\mathcal{S}_{\bar{\nu}} \cup \mathcal{F}_{\bar{\nu}})$ wherein all $x_k \in \mc{D} \implies x_{k+1} \not\in \mc{X}_{\bar{\nu}}$ with probability $1$.

In contrast, the set $\mc{X}_{1}$ is invariant for the system $x^+ = Ax + B {\kappa}_{N,1}(x)$. Therefore, the corresponding largest subset that is invariant for successful realizations of the variable $\nu_k\in\{0,1\}$ is
 \begin{equation*}
     \mathcal{S}_1 \coloneqq \{ x \in \mc{X}_{1} : Ax +  B {\kappa}_{N,1}(x) \in \mc{X}_{1}\} = \mc{X}_{1}.
 \end{equation*}
Therefore
    \begin{equation*}
\Po{Ax_k + \nu_k B \kappa_{N,1}(x_k)\in\mc{X}_{1}\given x_k=x} \geq 
\bn.
\end{equation*}
This completes the proof.
\end{proof}

\section{Illustrative examples}
\label{sec:illus}

Consider the system dynamics~\eqref{system} with
\begin{equation*}
    A=
\begin{bmatrix}
    1.2 & 1 \\ 0 &  1
\end{bmatrix}, \quad
B=
\begin{bmatrix}
    0.5 \\ 1
\end{bmatrix}, 
\end{equation*}
and state constraint set $\mathbb{X} = \{ x=(x_{1},x_2) : \abs{x_{1}} \leq 4.5, \abs{x_{2}} \leq 1.5 \}$. The critical arrival threshold for this system is determined to be $\nu_{c} \approx 0.31$ (to two d.p.), and we consider the scenario where $\bar{\nu}=0.7>\nu_{c}$.

We set  $Q= I_2$, $R=1$ and $\beta=1$. The terminal gain $K_{f,\bar{\nu}} = \begin{bmatrix}-0.5304, -1.1612\end{bmatrix}$ is the optimal one associated with the MARE, with corresponding $P_{\bn}=\left[\begin{smallmatrix} 5.2070 & 4.7835 \\ 4.7835 & 9.6280 \end{smallmatrix}\right]$, thus satisfies~\eqref{eq:ly2}. For optimistic MPC, the terminal gain $K_{f,1} = [-0.6302, -1.1009]$ is the optimal LQR gain, with $P_1=\left[\begin{smallmatrix} 3.3109 & 1.5931 \\ 1.5931 & 2.7647 \end{smallmatrix}\right]$, which satisfies~\eqref{eq:Plyap3}. We consider a horizon of $N=5$.

Firstly, for the input set $\mathbb{U}=\{u : \abs{u} \leq 1.5 \}$, Tab.~\ref{tab:values} compares the value functions of SMPC and OMPC at a state $x = \begin{bmatrix} -1 & -1.5 \end{bmatrix}^\top$ when the horizons are $N=5$ and $6$. This illustrates Theorem~\ref{thm:mono}; the SMPC value function increases with $N$. In fact, $\tilde{R}^1_{5,0} = 76.7507 = \tilde{R}^1_{6,1}$ but the $N=6$ problem imposes the extra input weight $\tilde{R}^1_{6,0} = 135.0247$ on the first input $u_0$.

\begin{table}[b]
    \centering
    \caption{Comparison of stochastic and optimistic value functions with horizons of $N=5$ and $6$, when $\mathbb{U}=\{u : |u| \leq 1.5 \}$ and $x = \begin{bmatrix} -1 & -1.5 \end{bmatrix}^\top$.}
    \begin{tabular}{rrrr}
    \toprule
    $V^s_{5,1}(x)$ & $V^s_{6,1}(x)$ & $V^s_{5,\bn}(x)$ & $V^s_{6,\bn}(x)$ \\
    \midrule
    $9.0449$ & $9.0449$ & $107.6320$ & $159.2154$\\
    \bottomrule
    \end{tabular}\label{tab:values}%
\end{table}

When the input constraint set is $\mathbb{U}=\{u : \abs{u} \leq 1.5 \}$, there is no ordering between the terminal sets $\mbb{X}_{f,\bn}$ and $\mbb{X}_{f,1}$ but $\mbb{X}_{5,\bn} \subset \mbb{X}_{5,1}$. This suggests an apparent feasibility advantage for optimistic MPC. However, in this case, the region $\mc{D}$ is empty (Fig.~\ref{barPU1}), implying that, under SMPC, all states in $\mbb{X}_{N,\bn}$ admit trajectories that remain feasible with non-zero probability. 

\begin{figure}[t]
\centering\footnotesize
\subfloat[SMPC ($\bn$)]{\includegraphics[width=0.45\linewidth]{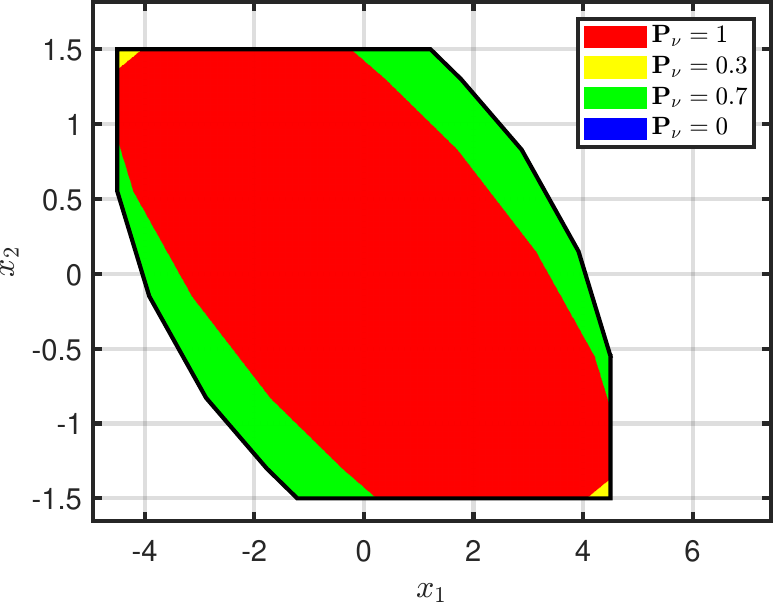}}\hfil
\subfloat[OMPC ($\bn=1$)]{\includegraphics[width=0.45\linewidth]{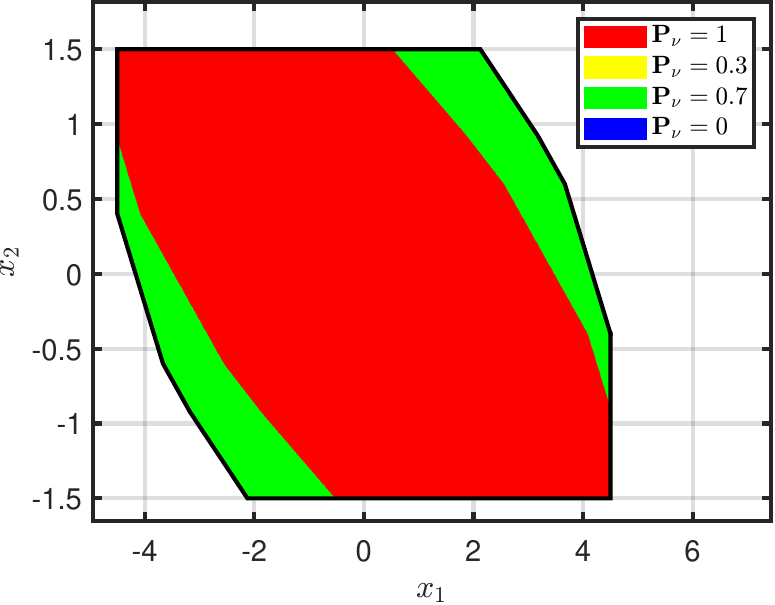}}
\caption{Characterization of the regions $\mc{S}_{\bn} \cap \mc{F}_{\bn}$ and $\mc{S}_{1} \cap \mc{F}_{1}$ (red), $\mc{S}_{\bn}\setminus \mc{F}_{\bn}$ and $\mc{S}_{1}\setminus \mc{F}_{1}$ (green), $\mc{F}_{\bn}\setminus \mc{S}_{\bn}$ and $\mc{F}_{1}\setminus \mc{S}_{1}$ (yellow), and $\mc{D} = \mbb{X}_{N,\bn} \setminus (\mc{S}_{\bn} \cup \mc{F}_{\bn})$ (blue) for the case that $\mbb{U} = \{|u| : {u} \leq 1.5\}$. The set $\mc{D}$, from which $x$ leaves $\mbb{X}_{N,\bn}$ with probability $1$, is empty.}
\label{barPU1}
\end{figure}

In contrast, when the input constraint set is $\mathbb{U}=\{u : \abs{u} \leq5 \}$, it can be verified that $\mbb{X}_{5,\bn} = \mbb{X}_{5,1}$. Therefore, neither formulation enjoys an inherent advantage in terms of feasibility. This removes feasibility as a confounding factor and allows the comparison to focus on the probabilistic behaviour induced by the respective control laws. Fig.~\ref{barPU5} illustrates the sets characterized in Theorem~\ref{thm:D} for this case. Notably, the set $\mc{D}$ is non-empty, meaning there exist states within $\mbb{X}_{5,\bn}$ whose successors, under stochastic MPC, leave the feasibility region with probability~$1$.

\begin{figure}[t]
\centering\footnotesize
\subfloat[SMPC ($\bn$)]{\includegraphics[width=0.45\linewidth]{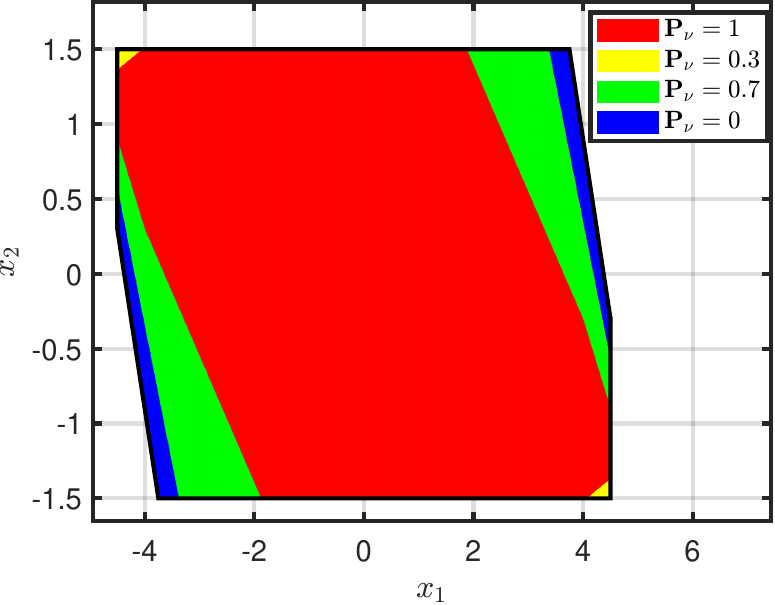}}\hfil
\subfloat[OMPC ($\bn=1$)]{\includegraphics[width=0.45\linewidth]{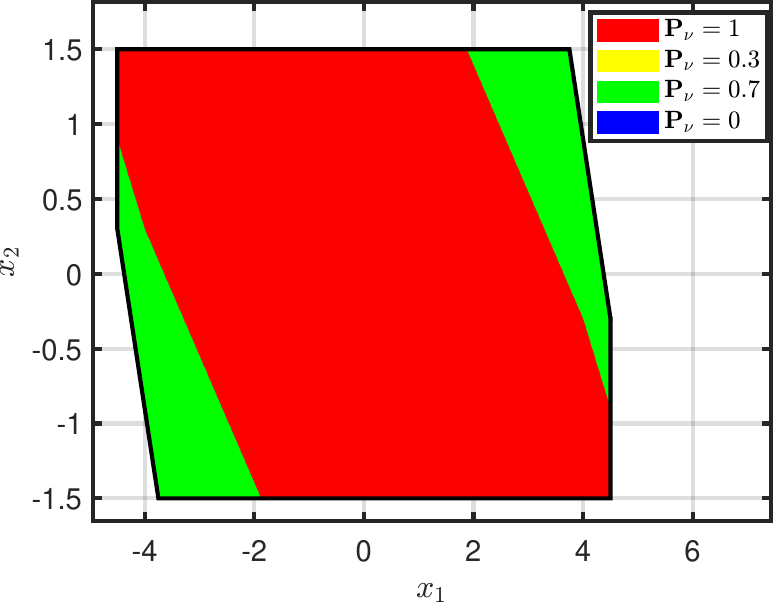}}
\caption{Characterization of the regions $\mc{S}_{\bn} \cap \mc{F}_{\bn}$ and $\mc{S}_{1} \cap \mc{F}_{1}$ (red), $\mc{S}_{\bn}\setminus \mc{F}_{\bn}$ and $\mc{S}_{1}\setminus \mc{F}_{1}$ (green), $\mc{F}_{\bn}\setminus \mc{S}_{\bn}$ and $\mc{F}_{1}\setminus \mc{S}_{1}$ (yellow), and $\mc{D} = \mbb{X}_{N,\bn} \setminus (\mc{S}_{\bn} \cup \mc{F}_{\bn})$ (blue) for the case that $\mbb{U} = \{|u| : {u} \leq 5\}$. The set $\mc{D}$, from which $x$ leaves $\mbb{X}_{N,\bn}$ with probability $1$, is non-empty.}
\label{barPU5}
\end{figure}

\section{Conclusion and future work}
\label{sec:conc}

This paper has studied the use of expectation-based stochastic MPC for linear constrained systems subject to multiplicative binary input uncertainty. We have shown that such formulations can exhibit structural properties that differ fundamentally from those of deterministic MPC, including the loss of monotonicity of the value function with respect to the prediction horizon and the potential deterioration of value function–based approximations of the region of attraction. Moreover, we have demonstrated that stochastic MPC may certify feasibility for states from which constraint violation occurs almost surely, revealing a fundamental disconnect between probabilistic constraint satisfaction and realized closed-loop feasibility.

 These results highlight intrinsic limitations of stochastic MPC formulations based on expected dynamics, in that (i) the controller lacks desirable properties that follow naturally in deterministic settings and (ii) the guarantees they provide do not reflect the behaviour of the realized system. The findings suggest that that care is required when interpreting stability and feasibility certificates derived from averaged models in the presence of binary multiplicative uncertainty.

It will be interesting to enquire in future work whether the results hold for other stochastic predictive controllers, including scenario-based approaches and those employing chance-constrained formulations.

\bibliographystyle{IEEEtran}
\bibliography{root}

\end{document}